\documentclass[11pt]{article}

\usepackage[margin=1.1in]{geometry}
\usepackage{graphicx}

\newcommand{\bft}{{\bf t}}

\newcommand{\z}{{\bf z}}

\usepackage{amssymb, amsmath, amsthm}
\usepackage{thmtools}

\newtheorem{theorem}{Theorem}[section]
\newtheorem{proposition}[theorem]{Proposition}
\newtheorem{corollary}[theorem]{Corollary}

\begin{document}

\title{Accelerated Polynomial Evaluation and Differentiation \\
at Power Series in Multiple Double Precision\thanks{Supported by 
the National Science Foundation under grant DMS 1854513.}}

\author{Jan Verschelde\thanks{University of Illinois at Chicago,
Department of Mathematics, Statistics, and Computer Science,
851 S. Morgan St. (m/c 249), Chicago, IL 60607-7045
Email: {\tt janv@uic.edu}, URL: {\tt http://www.math.uic.edu/$\sim$jan}.}}

\date{12 March 2021}


\maketitle

\begin{abstract}
The problem is to evaluate a polynomial in several variables and
its gradient at a power series truncated to some finite degree
with multiple double precision arithmetic.
To compensate for the cost overhead of multiple double precision
and power series arithmetic, data parallel algorithms 
for general purpose graphics processing units are presented.
The reverse mode of algorithmic differentiation is organized into 
a massively parallel computation of many convolutions and additions
of truncated power series.
Experimental results demonstrate that teraflop performance is obtained
in deca double precision with power series truncated at degree~152.
The algorithms scale well for increasing precision and increasing degrees. 

\medskip

\end{abstract}

\section{Introduction}

Solving systems of many polynomial equations in several variables
is needed in various fields of science and engineering.
Numerical continuation methods apply path trackers~\cite{Mor87}.
A path tracker computes approximations of the solution paths defined
by a family of polynomial systems.
The paths start at known solutions of easier systems
and end at the solutions of the given system.
The evaluation and differentiation of polynomials
often dominates the computational cost.

The main motivation for this paper is to accelerate a new robust
path tracker~\cite{TVV20a}, added recently to PHCpack~\cite{Ver99},
which requires power series expansions
of the solution series of polynomial systems.
As shown in~\cite{TVV20b}, double precision may no longer suffice
to obtain accurate results, for larger systems, and for longer power
series, truncated at higher degrees.
The double precision can be extended with double doubles, 
triple doubles, quad doubles, etc.,
applying multiple double arithmetic~\cite{MBDJJLMRT18}.
The goal is to compensate for the computational cost overhead 
caused by power series and multiple double precision by the
application of data parallel algorithms on general purposed
graphics processing units (GPUs).

The CUDA programming model (see~\cite{KH13} for an introduction)
is applied.  The software was developed on five different 
NVIDIA graphics cards: the C2050, K20C, P100, V100 (on Linux);
and the GeForce~2080 (on Windows).

\noindent {\bf Prior work.}  Data parallel algorithms 
for polynomial evaluation and differentiation~\cite{GW08}
were first presented by G.~Yoffe and the author in~\cite{VY12},
using double double and quad double arithmetic of~\cite{HLB01} on the host
and of~\cite{LHL10} on the device.
Adding accelerated linear algebra~\cite{VY13} led to an accelerated
Newton's method~\cite{VY14}, and to accelerated path 
trackers~\cite{VY15a,VY15b}.

\noindent {\bf Related work.}  The authors of~\cite{GM20a,GM20b}
describe a GPU implementation of a path tracker,
with an application to kinematic synthesis.
Instead of an adaptive step size control, paths for which the desired 
accuracy is not achieved are recomputed with a smaller step size.
The source code for the computations of~\cite{GM20a,GM20b} can be
found in the appendix of~\cite{Gla20}.

Accelerating the multiplication of polynomials is reported
in~\cite{Eme09} and~\cite{HM12,HLMMPX14,MP10}.
Those algorithms are applied with exact, modular arithmetic,
to polynomials with huge degrees.
A GPU-accelerated application of adjoint algorithmic differentiation
is implemented in~\cite{GHRKN16} for the gradient computation
of cost functions given by computer programs.

In~\cite{IK20}, several software packages for high precision arithmetic
on GPUs are considered.  For the problem of matrix-vector multiplication,
the double double arithmetic of CAMPARY~\cite{JMPT16} performs best.
In quad double precision, the performance of the implementation with CAMPARY
comes close to the multiple precision proposed by~\cite{IK20}.

\noindent {\bf Contributions.}
This paper extends the ideas of~\cite{VY12} to power series
and to more levels of multiple double precision,
with the code generated by the CAMPARY software~\cite{JMPT16}.
In addition to double, double double, and quad double precision,
the algorithms in this paper run also in triple, penta, octo, and
deca double precision, extending double precision respectively
three, five, eight, and ten times.

In~\cite{KH13}, convolutions and scans are explained as parallel patterns.
This paper presents novel data staging algorithms.
In one addition of two power series consecutive threads in the block
add consecutive coefficients of the series.
In one multiplication of two power series, the number of steps equals
the degree~$d$ at which the series are truncated.
The computations in the data parallel algorithms are defined by two
sequences of jobs.  The first sequence computes all multiplications,
for all monomials, while the second sequence encodes the additions 
of all evaluated monomials.  
The theoretical speedup of the novel parallel algorithms is a multiple
of~$d$.

In deca double precision, on power series truncated at degree 152,
teraflop performance is reached on the P100 and the V100.
Experimental results show the scalability for increasing degrees
and increasing precisions.

\section{Convolutions}

Consider the product $z$ of two power series $x$ and $y$,
both truncated to the same degree~$d$,
so the input consists of two sequences of $d+1$ coefficients.
The output are the $d+1$ coefficients of the product~$z$.
In a data parallel algorithm with $d+1$ threads,
thread~$k$ will compute the $k$-th coefficient $z_k$ of~$z$,
defined by the formula 
\begin{equation} \label{eqconvform}
   z_k = \sum_{i=0}^k x_i y_{k-i}, \quad k = 0, 1, \ldots, d,
\end{equation}
where $x_i$ is the $i$-th coefficient of~$x$
and $y_{k-i}$ is the $(k-i)$-th coefficient of~$y$.
In the direct application of the formula in~(\ref{eqconvform}),
every thread performs a different number of computations,
which results in thread divergence.

The remedy for this thread divergence is to insert zero numbers
before the second vector when threads load the numbers into
the shared memory of the block.  In the statements below,
$X$, $Y$, and $Z$ represent the shared memory locations,
respectively for the coefficients of $x$, $y$, and $z$.
The $Y$ has space for at least $2 d + 2$ numbers.
In the data parallel algorithm with zero insertion,
thread~$k$ executes the following statements, expressed in pseudo code:

\begin{tabbing}
\hspace{1cm} \= 1. $X_k := x_k$ \\
             \> 2. $Y_k := 0$ \\
             \> 3. $Y_{d+k} := y_k$ \\
             \> 4. $Z_k := X_0 Y_{d+k}$ \\
             \> 5. for\= ~$i$ from 1 to $d$ do 
 $Z_k := Z_k + X_i Y_{d+k-i}$ \\
             \> 6. $z_k := Z_k$
\end{tabbing}
Thread $k$ executes the same statements on different data.
In the statements above, the $k=0, 1, \ldots, d$
of formula~(\ref{eqconvform}) is implicit.
As the data parallel version eliminates the outside loop on~$k$,
it is expected to run about $d$ times faster
than the sequential application of formula~(\ref{eqconvform}).

The zero insertion justifies the auxiliary vector~$Y$,
but why are $X$ and~$Z$ needed?  The coefficients $x$, $y$, and~$z$
reside in the global memory of the device.
Access to global memory is slower than access to shared memory.
Observe that all threads need access to $x_0$.
Retrieving $x_0$ once from global memory and then $d+1$ times
from shared memory is expected to be faster than $d+1$ times 
retrieving $x_0$ from global memory.
The same argument applies to the assignments to $Z_k$.
Every thread assigns once to $Z_k$ and then updates $Z_k$
as many as $d$ times.  Only at the very end of the algorithm
is the value of $Z_k$ in shared memory assigned to the value
of $z_k$ in global memory.
Another benefit of using vectors in shared memory occurs
when one needs to update the same series $x$ or $y$ with the product.
Replacing $z$ by $x$ or $y$ in~(\ref{eqconvform})
requires an auxiliary vector.

All coefficients are stored in consecutive memory locations,
promoting efficient memory access.
For arrays of doubles, threads with consecutive indices
access the corresponding consecutive memory locations.
For complex numbers and multiple double numbers,
the same efficient memory access is obtained by storing
real and imaginary parts of complex numbers in separate arrays
and by storing all parts of multiple double numbers 
in separate arrays.

The other basic operation is the addition of two power series.
In the data parallel version, one block of threads adds two power series.
If the block has exactly has exactly as many threads as the number
of coefficients of the power series, then thread $k$ adds the
$k$-th coefficient of the two series.

The next two sections elaborate the scheduling of convolution jobs.

\section{Monomial Evaluation and Differentiation}

Consider a monomial $a ~\! x_1 x_2 \cdots x_n$,
in the $n$ variables $x_1$, $x_2$, $\ldots$, $x_n$,
and $a$ is a nonzero power series, truncated to degree~$d$.
We want to evaluate and differentiate this monomial
at a sequence of $n$ power series~$\z = (z_1, z_2, \ldots, z_n)$.
All series in $\z$ are also truncated to degree~$d$.

For monomials with positive powers, e.g.: $x_1^3 x_2^5$,
observe that the value of $x_1^2 x_2^4$ is not only 
a factor of the monomial value, but is also a factor 
in all values of the derivatives.
Therefore, we write $x_1^3 x_2^5$ as $a~\! x_1 x_2$,
where $a = x_1^2 x_2^4$.  This common factor is then
evaluated with a table of powers of the variables.

The statements below assume that $n$ is larger than 2.
Each $\star$ represents a convolution of two power series.
The $n$ forward products are stored in the $n$-dimensional array~$f$.
The $(n-2)$-dimensional array~$b$ collects the backward products.
Other partial derivatives can then be found 
in the $(n-2)$-dimensional array $c$ of cross products.
\begin{tabbing}
\hspace{1cm} \= 
    1. $f_1 := a \star z_1$ \\
 \> 2. for $j$ from 2 to $n$ do $f_j := f_{j-1} \star z_i$ \\
 \> 3. $b_1 := z_n \star z_{n-1}$ \\
 \> 4. for $j$ from 2 to $n-2$ do $b_j := b_{j-1} \star z_{n-j}$ \\
 \> 5. $b_{n-2} := b_{n-2} \star a$ \\
 \> 6. for $j$ from 1 to $n-3$ do $c_j := f_j \star b_{n-3-j}$ \\
 \> 7. $c_{n-2} := f_{n-2} \star z_n$
\end{tabbing}
The amount of $\star$ operations equals the total amount of
auxiliary storage, plus one (as $b_{n-2}$ gets assigned twice): $3n - 3$.
In the example below for $n=5$, if one expands
the left and right operands of $\star$ into their values,
one can verify that $b_3$, $c_1$, $c_2$, $c_3$, and $f_4$
contain the values of all five partial derivatives.
The organization in three columns shows the parallelism.
Statements on the same line in~(\ref{eqstatements})
can be executed in parallel.

\begin{equation} \label{eqstatements}
  \begin{array}{lll}
   f_1 := a \star z_1 
 & b_1 := z_5 \star z_4 \\
   f_2 := f_1 \star z_2 
 & b_2 := b_1 \star z_3 \\ 
   f_3 := f_2 \star z_3 
 & b_3 := b_2 \star z_2 
 & c_1 := f_1 \star b_2 \\ 
   f_4 := f_3 \star z_4 
 & b_3 := b_3 \star a 
 & c_2 := f_2 \star b_1 \\ 
   f_5 := f_4 \star z_5 & 
 & c_3 :=  f_3 \star z_5 
  \end{array}
\end{equation}

For $n = 5$, three blocks of threads may evaluate and differentiate
one monomial in five steps.
For $n > 5$, observe that $c_1$ needs $z_n z_{n-2} \cdots z_3$,
or the value of $b_{n-3}$.
While this observation seems to limit the amount of parallelism,
the cross products need not be computed one after the other.

\begin{proposition} \label{propmon}
The $j$-th cross product $c_j$ can be computed after $\max(j, n-3-j)$ steps.
\end{proposition}
\begin{proof}
For $n=3$, the only cross product is $c_1 := f_1 \star z_3$,
and $c_1$ can be computed after $f_1$ has been computed in the first step.
After one step, $c_1$ can be computed for~$n=3$.

For $n > 3$, consider $c_{n-2} := f_{n-2} \star z_n$.
The computation of $c_{n-2}$ has to wait for the computation of $f_{n-2}$,
which requires $n-2$ steps.
For $j < n-2$, $c_j := f_j \star b_{n-3-j}$ and $c_j$ can be computed
after $f_j$ and $b_{n-3-j}$ have been computed, which each take respectively
$j$ and $n-3-j$ steps.
Thus, after $\max(j, n-3-j)$ steps, $c_j$ can be computed.
\end{proof}

\begin{corollary} \label{cormon}
Given sufficiently many blocks of threads,
monomial evaluation and differentiation takes $n$ steps for $n$ variables.
\end{corollary}

As the number of convolutions to evaluate and differentiate one monomial
in $n$ variables equals $3n - 3$,
Corollary~\ref{cormon} means that $3d$ is the upper bound on the speedup,
where $d$ is the truncation degree of the power series.

\section{Polynomial Evaluation and Differentiation}

We consider a polynomial $p$ of $N$ monomials
in $n$ variables with nonzero coefficients
as power series, all truncated to the same degree~$d$.
The coefficient of the $k$-th monomial is denoted as $a_k$
and $n_k$ variables appear in the monomial,
for $k$ ranging from 1 to $N$.
The variables in the $k$-th monomial
are defined by the tuple of indices $(i_1, i_2, \ldots, i_{n_k})$
with $1 \leq i_1 < i_2 < \cdots < i_{n_k} \leq n$.
We want to evaluate and differentiate
\begin{equation} \label{eqpolynomial}
   p(x_1,x_2,\ldots,x_n)
   = a_0 + \sum_{k=1}^N a_k ~\! x_{i_1} x_{i_2} \cdots x_{i_{n_k}},
\end{equation}
at a sequence of $n$ power series~$\z = (z_1, z_2, \ldots, z_n)$.
All series in $\z$ are also truncated to degree~$d$.
The constant term~$a_0$ is not included in the count~$N$.

In the data parallel evaluation and differentiation,
every monomial has three separate arrays 
of forward, backward, and cross products, denoted respectively
by $f$, $b$, and~$c$.
In the example below, the first index of $f$, $b$, and $c$
corresponds to the monomial index.

\begin{equation} \label{eqexpoly}
 \begin{array}{rccccccc}
   p & = ~ a_0~+~ 
         & a_1 x_1 x_3 x_6 & + & a_2 x_1 x_2 x_5 x_6 & + & a_3 x_2 x_3 x_4 \\
     &   & f_{1,1} := a_1 \star z_1 
     &   & f_{2,1} := a_2 \star z_1 
     &   & f_{3,1} := a_3 \star z_2 \\
     &   & f_{1,2} := f_{1,1} \star z_3 
     &   & f_{2,2} := f_{2,1} \star z_2 
     &   & f_{3,2} := f_{3,1} \star z_3 \\
     &   & f_{1,3} := f_{1,2} \star z_6 
     &   & f_{2,3} := f_{2,2} \star z_5 
     &   & f_{3,3} := f_{3,2} \star z_4 \\
     &   & 
     &   & f_{2,4} := f_{2,3} \star z_6
     &   &  \\
     &   & b_{1,1} := z_6 \star z_3 
     &   & b_{2,1} := z_6 \star z_5 
     &   & b_{3,1} := z_4 \star z_3 \\
     &   & b_{1,1} := b_{1,1} \star a_1 
     &   & b_{2,2} := b_{2,1} \star z_2
     &   & b_{3,1} := b_{3,1} \star a_3 \\
     &   & 
     &   & b_{2,2} := b_{2,2} \star a_2
     &   &  \\
     &   & c_{1,1} := f_{1,1} \star z_6 
     &   & c_{2,1} := f_{2,1} \star b_{2,1}
     &   & c_{3,1} := f_{3,1} \star z_4 \\
     &   & 
     &   & c_{2,2} := f_{2,2} \star z_6
 \end{array}
\end{equation}
The 21 convolutions in~(\ref{eqexpoly} are arranged in~(\ref{eqparallel}).
Statements on the same line can be computed in parallel.
\begin{equation} \label{eqparallel}
  \begin{array}{ccccccccc}
     f_{1,1} & b_{1,1} & & f_{2,1} & b_{2,1} & & f_{3,1} & b_{3,1} \\
     f_{1,2} & b_{1,1} & c_{1,1} & f_{2,2} & b_{2,2} & c_{2,1}
   & f_{3,2} & b_{3,1} & c_{3,2} \\
     f_{1,3} & & & f_{2,3} & b_{2,2} & c_{2,2} & f_{3,3} & & \\
      & & & f_{2,4} & & & & & \\
  \end{array}
\end{equation}
If 9 thread blocks are available, all 21 convolutions can
be computed in 4 steps.  The value of $p$ at $\z$ and all six
partial derivatives are listed below:
\begin{equation}
  a_0 + f_{1,3} + f_{2,4} + f_{3,3}, \quad
  b_{1,1} + b_{2,1}, \quad
  c_{2,1} + b_{3,1}, \quad
  c_{1,1} + c_{3,1}, \quad
  f_{3,2}, \quad
  c_{2,2}, \quad
  f_{1,2} + f_{2,3}.
\end{equation}
If 7 threads blocks are available, all values can be computed in two steps. 

Two steps are needed for the value of $p$.  The first step computes
$f_{1,3} := a_0 + f_{1,3}$ and $f_{3,3} := f_{2,4} + f_{3,3}$ simultaneously.
The second step then does $f_{3,3} := f_{3,3} + f_{1,3}$,
so the value of~$p$ is in $f_{3,3}$.

Obviously, as convolutions for different monomials can be computed
in parallel, the result in Corollary~\ref{cormon} extends directly
to polynomials.
\begin{corollary} \label{corpol}
Consider a polynomial~$p$ in $n$ variables, with $N$ monomials.
Let $m$ be the number of variables in that monomial of~$p$ that has the
largest number of variables.
Given sufficiently many blocks of threads,
the evaluation and differentiation of~$p$ takes 
$m + \lceil ~\! \log_2(N) ~\! \rceil$ steps.
\end{corollary}

To interpret Corollary~\ref{corpol}, assume every monomial has $m$
variables.  Then the total number of convolutions equals~$N (3m-3)$
and the number of additions equals~$N$.
As in the case of one monomial, the speedup factor of~$3d$ is present.
For polynomials with many monomials, for $N \gg m$,
the upper bound on the speedup is $d N/\log_2(N)$.

\section{Accelerated Evaluation and Differentiation}

The accelerated polynomial evaluation and differentiation algorithm
proceeds in two stages.  The first stage computes all convolutions.
The second stage adds up the evaluated and differentiated monomials.

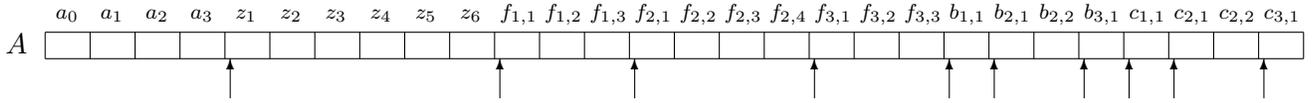
\begin{figure}[t!]
\begin{center}
\begin{picture}(476,40)(0,0)
\put(-15,17){$A$}
\put(4,30){\scriptsize $a_0$}
\put(21,30){\scriptsize $a_1$}
\put(38,30){\scriptsize $a_2$}
\put(55,30){\scriptsize ${a_3}$}
\put(72,30){\scriptsize ${z_1}$}
\put(89,30){\scriptsize ${z_2}$}
\put(106,30){\scriptsize ${z_3}$}
\put(123,30){\scriptsize ${z_4}$}
\put(140,30){\scriptsize ${z_5}$}
\put(157,30){\scriptsize ${z_6}$}
\put(172,30){\scriptsize $f_{1,1}$}
\put(189,30){\scriptsize $f_{1,2}$}
\put(206,30){\scriptsize $f_{1,3}$}
\put(223,30){\scriptsize $f_{2,1}$}
\put(240,30){\scriptsize $f_{2,2}$}
\put(257,30){\scriptsize $f_{2,3}$}
\put(274,30){\scriptsize $f_{2,4}$}
\put(291,30){\scriptsize $f_{3,1}$}
\put(308,30){\scriptsize $f_{3,2}$}
\put(325,30){\scriptsize $f_{3,3}$}
\put(342,30){\scriptsize $b_{1,1}$}
\put(359,30){\scriptsize $b_{2,1}$}
\put(376,30){\scriptsize $b_{2,2}$}
\put(393,30){\scriptsize $b_{3,1}$}
\put(410,30){\scriptsize $c_{1,1}$}
\put(427,30){\scriptsize $c_{2,1}$}
\put(444,30){\scriptsize $c_{2,2}$}
\put(461,30){\scriptsize $c_{3,1}$}
\put(0,25){\line(1,0){476}}
\put(0,15){\line(1,0){476}}
\put(0,15){\line(0,1){10}}
\put(17,15){\line(0,1){10}}
\put(34,15){\line(0,1){10}}
\put(51,15){\line(0,1){10}}
\put(68,15){\line(0,1){10}}
\put(85,15){\line(0,1){10}}
\put(102,15){\line(0,1){10}}
\put(119,15){\line(0,1){10}}
\put(136,15){\line(0,1){10}}
\put(153,15){\line(0,1){10}}
\put(170,15){\line(0,1){10}}
\put(187,15){\line(0,1){10}}
\put(204,15){\line(0,1){10}}
\put(221,15){\line(0,1){10}}
\put(238,15){\line(0,1){10}}
\put(255,15){\line(0,1){10}}
\put(272,15){\line(0,1){10}}
\put(289,15){\line(0,1){10}}
\put(306,15){\line(0,1){10}}
\put(323,15){\line(0,1){10}}
\put(340,15){\line(0,1){10}}
\put(357,15){\line(0,1){10}}
\put(374,15){\line(0,1){10}}
\put(391,15){\line(0,1){10}}
\put(408,15){\line(0,1){10}}
\put(425,15){\line(0,1){10}}
\put(442,15){\line(0,1){10}}
\put(459,15){\line(0,1){10}}
\put(476,15){\line(0,1){10}}
\put(70,0){\vector(0,1){15}}  
\put(172,0){\vector(0,1){15}} 
\put(223,0){\vector(0,1){15}} 
\put(291,0){\vector(0,1){15}} 
\put(342,0){\vector(0,1){15}} 
\put(359,0){\vector(0,1){15}} 
\put(393,0){\vector(0,1){15}} 
\put(410,0){\vector(0,1){15}} 
\put(427,0){\vector(0,1){15}} 
\put(461,0){\vector(0,1){15}} 
\end{picture}
\caption{The data array used to compute 
the forward, backward, and cross products
to evaluate a polynomial and its gradient 
of the example in~(\ref{eqexpoly}).
Every box represents $d+1$ doubles for the coefficients of a series
truncated at degree~$d$.  The arrows point at the start position of
the input series and at the space for the forward, backward, cross products 
for every monomial. }
\label{figdatapoly}
\end{center}
\end{figure}

As the number of threads in each block matches the number of
coefficients in each truncated power series,
within each block the natural order of the data follows
the coefficient vectors of the power series.
In preparation for the launching of the kernels 
the staging of the data must be defined.

Extracting the data structures of~(\ref{eqpolynomial}),
the input of the algorithm consists of the following:
\begin{enumerate}
\item $N$, the number of monomials; 
\item $n$, the number of variables;
\item $d$, the degree at which all series are truncated;
\item $a_k$, truncated series as the coefficient of the $k$-th monomial,
      $k = 1, 2, \ldots, N$;
\item $(i_1, i_2, \ldots, i_{n_k})$,
      indices of the variables in the $k$-th monomial,
      where $1 \leq i_1 < i_2 < \cdots < i_{n_k} \leq n$,
      for $k = 1,2, \ldots, N$; 
\item $\z = (z_1, z_2, \ldots, z_n)$,
      $n$ power series truncated to degree~$d$.
\end{enumerate}
The output of the first stage are $N$ tuples
$(f_k, b_k, c_k)$, for $k=1,2,\ldots,N$,
where $f_{k,j}$ are the forward products, $j=1,2,\ldots,n_k$,
$b_{k,j}$ are the backward products,
$j=1,2,\ldots, \max(1, n_k-2)$,
$c_{k,j}$ are the cross products, $j=1,2,\ldots, n_k-2$.
The $\max(1, n_k-2)$ in the upper bound for $b_{k,j}$ is
for the special case $n_k=2$, to store $z_{i_2} \star a_2$.

The output of the first stage is the input of the second stage.
The second stage adds for each monomial $k$, the last forward product
$f_{k, n_k}$ to obtain $p(\z)$.
The values of the derivatives of the $k$-th monomial are in
$f_{k,n_k-1}$, $b_{k,n_k-2}$, and $c_{k,j}$.

The data parallel algorithm to compute all convolutions is defined
by the data layout.
The total count of numbers involved in all convolution 
and addition jobs is
\begin{equation} \label{eqcountnum}
  \begin{array}{l}
  {\displaystyle
   e = (d+1)
   \left( 1 + N + n
     + \sum_{k=1}^N \left( n_k 
     \phantom{\sum_{k=1}^N} \!\!\!\!\!\!\!\!
              + \max(1, n_k-2) + \max(0, n_k-2) 
\phantom{\sum_{k=1}^N} \!\!\!\!\!\!\!\!
\right)
   \right).}
  \end{array}
\end{equation}
The first factor~$(d+1)$ in~(\ref{eqcountnum}) counts the number
of coefficients in all series truncated to degree~$d$.
The five terms in the second factor in~(\ref{eqcountnum}) count
respectively the constant coefficient $a_0$,
the $N$ coefficients $a_k$, the $n$ input series in $\z$,
the $n_k$ forward, the $\max(1,n_k-2)$ backward, and
the $\max(0,n_k-2)$ cross products.

Figure~\ref{figdatapoly} illustrates the layout of the data vector
for the polynomial in~(\ref{eqexpoly}).

The data are in an array $A$ of $e$ doubles.
The order of the numbers in~$A$ follows the count as in~(\ref{eqcountnum}):
the coefficients $a_0$, $a_k$, and series $\z$ are followed by the numbers
in the forward, backward, and cross products.
Each job is then characterized by a triplet of indices in this~$A$.
The first two indices point respectively at the start of the first and
the second input.  The third index in the triplet defines the start
of the output.   For example, the triplet for the convolution
$f_{1,1} := a_1 \star z_1$ for the polynomial in~(\ref{eqexpoly})
is $(d+1, 4d+4, 10d+10)$, for degree $d$, 
as the coefficients for $a_1$ start after
the first $d+1$ coefficients for $a_0$, $z_1$ starts after the the
first four series, and $f_{1,1}$ after the ten series that define
the coefficients of the polynomials and the input.

For complex numbers, the data are in two arrays, one for the real parts
and the other for the imaginary parts.  For $m$-fold double numbers,
there are $m$ data arrays, all following the same layout as for~$A$,
described above.

Jobs are placed in layers, following the lines of~(\ref{eqparallel})
for the example polynomial.
Jobs to compute $f_{k,j}$ and $b_{k,j}$ are at layer~$j$.
Following Proposition~\ref{propmon},
the layer of $c_{k,j}$ is $\max(j, n_k-3-j) + 1$.
All jobs in the same layer can be executed at the same time.
A kernel is launched with as many blocks as the number of jobs in one layer.
One convolution job is executed by one block of threads.
In addition to the data array~$A$, the kernel is launched with 
the triple of coordinates which define each job.
Each block of threads extracts the triplet according to its block number.

The coordinates for each convolution job in the first stage
of the algorithm depend only on the structure of the monomials 
and are computed only once.
The same holds for the coordinates of the addition jobs in the second
stage of the algorithm.  Each addition job updates one series with another,
so one pair of indices defines one addition job.
For example, for the polynomial in~(\ref{eqexpoly}),
the first update $f_{1,3} := a_0 + f_{1,3}$ has coordinates $(0, 12d + 12)$,
as $a_0$ comes first and $f_{1,3}$ is positioned after 12 series in $A$,
see Figure~\ref{figdatapoly}.

For a convolution job $j$, 
let $\bft = (t_1(j), t_2(j), t_3(j))$ denote the triplet
where $t_1(j)$ and $t_2(j)$ respectively are the locations in the
data array~$A$ of the first and second input, and
where $t_3(j)$ is the location in~$A$ of the output.
Given $A$ and $\bft$, we can then symbolically summarize the code
for the kernel to compute all convolution jobs at the same layer as
\begin{enumerate}
\item $B := \mbox{\tt blockIdx.x}$
\item $(i_1,i_2,i_3) :=  (t_1(B), t_2(B), t_3(B))$
\item $A[i_3:i_3\!+\!d\!+\!1]
      := A[i_1:i_1\!+\!d\!+\!1] \star A[i_2:i_2\!+\!d\!+\!1]$
\end{enumerate}
where the block index $B$ corresponds to the index $j$ of the convolution job,
and $[i\!:\!i\!+\!d\!+\!1]$ denotes the range of the coefficients in~$A$.
Likewise, for an addition job $j$,
let $\bft = (t_1(j), t_2(j))$ be the pair of input and update locations
in the data array~$A$.  Given $A$ and $\bft$, the kernel to execute
all addition jobs at the same layer is then summarized as
\begin{enumerate}
\item $B := \mbox{\tt blockIdx.x}$
\item $(i_1, i_2) := (t_1(B),t_2(B))$
\item $A[i_2:i_2\!+\!d\!+\!1]
      := A[i_2:i_2\!+\!d\!+\!1] + A[i_1:i_1\!+\!d\!+\!1]$
\end{enumerate}

The data staging algorithm to define the convolution jobs runs through 
the steps to compute all forward $f_{k,\ell}$, backward $b_{k,\ell}$, 
and cross products $c_{k,\ell}$,
for $k$ ranging from 1 to~$N$.  The $k$-th monomial has $n_k$ variables,
with indices $(i_1,i_2,\ldots,i_{n_k})$.  
\begin{equation}
   \alpha_k = \sum_{\ell = 1}^{k-1} n_\ell, \quad
   \beta_k = \alpha_{N+1} + \sum_{\ell = 1}^{k-1} \max(1,n_\ell-2),
   \mbox{and} \quad
   \gamma_k = \beta_{N+1} + \sum_{\ell = 1}^{k-1} \max(0,n_\ell-2)
\end{equation}
mark the positions respectively of $f_{k,1}$, $b_{k,1}$,
and $c_{k,1}$ in the data array~$A$.
For an example, see the arrows in Figure~\ref{figdatapoly}.
Let $J_\ell$ denote a set of convolutions jobs at level~$\ell$.
Jobs in $J_\ell$ can be executed after $\ell-1$ steps.
The simplified pseudo code below assumes all $n_k > 2$.

\noindent For $k$ from 1 to $N$ do
\begin{enumerate}
\item to execute $f_{k,1} := a_k \star z_{i_1}$:

      ~~~~$t_1 := k(d\!+\!1)$

      ~~~~$t_2 := (1+N+i_1-1)(d\!+\!1)$

      ~~~~$t_3 := (1+N+n+\alpha_k)(d\!+\!1)$

      ~~~~$J_1 := J_1 \cup \{ (t_1, t_2, t_3) \}$

\item for $\ell$ from 2 to~$n_k$ do 

      ~~~~to execute $f_{k,\ell} := f_{k,\ell-1} \star z_{i_\ell}$:

      ~~~~$t_1 := (1+N+n+\alpha_k+\ell-2)(d\!+\!1)$

      ~~~~$t_2 := (1+N+i_\ell-1)(d\!+\!1)$

      ~~~~$t_3 := (1+N+n+\alpha_k+\ell-1)(d\!+\!1)$

      ~~~~$J_\ell := J_\ell \cup \{ (t_1,t_2,t_3) \}$
\item to execute $b_{k,1} := z_{n_k} \star z_{n_k-1}$:

      ~~~~$t_1 := (1+N+n_k-1)(d\!+\!1)$

      ~~~~$t_2 := (1+N+n_k-2)(d\!+\!1)$

      ~~~~$t_3 := (1+N+n+\beta_k)(d\!+\!1)$

      ~~~~$J_1 := J_1 \cup \{ (t_1, t_2, t_3) \}$
\item for $\ell$ from 2 to $n_k-2$ do

      ~~~~to execute $b_{k,\ell} := b_{k,\ell-1} \star z_{n_k-\ell}$:

      ~~~~$t_1 := (1+N+n+\beta_k+\ell-2)(d\!+\!1)$

      ~~~~$t_2 := (1+N+n_k-\ell)(d\!+\!1)$

      ~~~~$t_3 := (1+N+n+\beta_k+\ell-1)(d\!+\!1)$

      ~~~~$J_\ell := J_\ell \cup \{ (t_1, t_2, t_3) \}$
\item to execute $b_{k,n_k-2} := b_{k,n_k-2} \star a_k$:

      ~~~~$t_1 := (1+N+n+\beta_k+n_k-3)(d\!+\!1)$

      ~~~~$t_2 := k(d\!+\!1)$

      ~~~~$t_3 := (1+N+n+\beta_k+n_k-3)(d\!+\!1)$

      ~~~~$J_{n_k-2} := J_{n_k-2} \cup \{ (t_1, t_2, t_3) \}$

\item for $\ell$ from 1 to $n_k-3$ do

      ~~~~to execute $c_{k,\ell} := f_{k,\ell} \star b_{k,n_k-3-\ell}$

      ~~~~$t_1 := (1+N+n+\alpha_\ell+\ell-1)(d\!+\!1)$

      ~~~~$t_2 := (1+N+n+\beta_k+n_k-3-\ell-1)(d\!+\!1)$
 
      ~~~~$t_3 := (1+N+n+\gamma_k+\ell-1)(d\!+\!1)$

      ~~~~$L := \max(\ell, n_k-3-\ell)$

      ~~~~$J_L := J_L \cup \{ (t_1, t_2, t_3) \}$

\item to compute $c_{k,n_k-2} := f_{k,n_k-2} \star z_{n_k}$:

      ~~~~$t_1 := (1+N+n+\alpha_k+n_k-3)(d\!+\!1)$

      ~~~~$t_2 := (1+N+n_k-1)(d\!+\!1)$

      ~~~~$t_3 := (1+N+n+\gamma_k+n_k-3)(d\!+\!1)$

      ~~~~$J_{n_k-2} := J_{n_k-2} \cup \{ (t_1, t_2, t_3) \}$

\end{enumerate}
Sets are natural data structures in the mathematical description
of the data staging algorithm, as the jobs in one $J_\ell$ can be
executed in any order.  In the implementation, the jobs in the same
layer are stored in three integer arrays.  
The block index $B$ is then used to get the coordinates
of the convolution job performed by block~$B$.

The pairs of indices for the addition jobs are defined
in a recursive manner, following the order of the tree summation algorithm.
To compute the value of the polynomial, 
add the forward product $f_{k,n_k}$ of the $k$-th monomial.
For stride~$L$ and level~$\ell$, apply the following:

~~~~to execute $f_{k,n_k} := f_{k,n_k} + f_{k-L,n_{k-L}}$:

~~~~~~~~$t_1 := 1 + N + n + \alpha_{k-L} + n_{k-L} - 1$

~~~~~~~~$t_2 := 1 + N + n + \alpha_k + n_k - 1$

~~~~~~~~$J_\ell := J_\ell \cup \{ (t_1, t_2) \}$

\noindent recursively, 
starting at $L = \lfloor N/2 \rfloor$ 
and level $\ell = \log_2(N)$ (assuming $N = 2^\ell$),
dividing $L$ by two in each step, until $L = 1$.
For $k = L$ in the formula
$f_{k,n_k} := f_{k,n_k} + f_{k-L,n_{k-L}}$,
replace $f_{k-L,n_{k-L}}$ by $a_0$.
The same recursive formula is applied to sum the first backward products
and all cross products to obtain the gradient.

\section{Computational Results}

\subsection{Equipment and Test Polynomials}

Table~\ref{tabgpus} summarizes the characteristics of each GPU,
with the focus on the core counts and the processor speeds,
because the problem is compute bound.
The first four GPUs in Table~\ref{tabgpus} are housed in a Linux
workstation, running CentOS.  The fifth GPU resides in a Windows laptop.

\begin{table}[hbt]
\begin{center}
\begin{tabular}{r||r|r|r|r|r||r}
  NVIDIA GPU~~~ &  CUDA  & \#MP  & \#cores/MP & \#cores & GHz 
 & host CPU GHz~~ \\ \hline
      Tesla C2050 & 2.0~~ &  14~~ &    32~~~~~ &   448~ & 1.15 
 & Intel X5690 3.47 \\
      Kepler K20C & 3.5~~ &  13~~ &   192~~~~~ &  2496~ & 0.71 
 & Intel E5-2670 2.60 \\
      Pascal P100 & 6.0~~ &  56~~ &    64~~~~~ &  3584~ & 1.33 
 & Intel E5-2699 2.20 \\
      Volta V100 & 7.0~~ &  80~~ &    64~~~~~ &  5120~ & 1.91 
 & Intel W2123 3.60 \\
 GeForce RTX 2080 & 7.5~~ &  46~~ &    64~~~~~ &  2944~ & 1.10 
 & Intel i9-9880H 2.30
\end{tabular}
\caption{The columns list the CUDA capability, 
the number of multiprocessors, the number of cores per multiprocessor,
the total number of cores, and the GPU clock rate.
For every GPU, its host CPU is listed with its clock rate,
and the host processor.}
\label{tabgpus}
\end{center}
\end{table}

The software was developed on the five GPUs listed in Table~\ref{tabgpus},
compiled with {\tt nvcc -O3} on the device,
with the code for the host compiled by {\tt gcc -O3} on the linux computers, 
and the community edition of Microsoft Visual Studio on the Windows laptop.
While running the same software on all five GPUs is obviously convenient,
more advanced features of newer devices are not utilized.
The main importance for the evaluation of our software is that the
unfair comparison with the CPU is avoided.
Taking into the account the double peak performance of the P100
and the V100 (4.7 TFLOPS and 7.9 TFLOPS respectively),
we may expect the V100 to be about 1.68 times faster than the P100.

The first test polynomial $p_1$ is a function of 16 variables.
Its evaluation adds to the constant term all 1,820 monomials
that are the products of exactly four variables.
The evaluation requires 16,380 convolutions and 9,084 additions.  
As each monomial has no more than four variables,
the 16,380 convolutions are performed in four kernel launches of
respectively 3,640, 5,460, 5,460, and 1,820 blocks.
The execution of the 9,084 additions requires 11 kernel launches of
respectively 4,542, 2,279, 1,140, 562, 281, 140, 78, 39, 20, 2,
and 1 blocks.  The second test polynomial $p_2$ is constructed
to require many more convolutions than additions,
respectively 24,192 versus 8,192.  To evaluate the third polynomial~$p_3$,
as many convolutions as additions are required: 24,256.

Table~\ref{tabtestpols} lists the characteristics of the test polynomials.
Compared to~$p_1$, 
$p_2$ has fewer monomials but each monomial has many more variables;
whereas $p_3$ has many more monomials, but each monomial has only
two variables.

\begin{table}[hbt]
\begin{center}
\begin{tabular}{r|rrr|rr}
        & \multicolumn{1}{c}{$n$} & \multicolumn{1}{c}{$m$} & 
          \multicolumn{1}{c|}{$N$} & \multicolumn{1}{c}{\#cnv} & 
          \multicolumn{1}{c}{\#add}  \\ \hline
  $p_1$ &  16 &  4 & 1,820 & 16,380 & 9,084 \\ 
  $p_2$ & 128 & 64 &   128 & 24,192 & 8,192 \\
  $p_3$ & 128 &  2 & 8,128 & 24,256 & 24,256 \\
\end{tabular}
\caption{For each polynomial, $n$ is the total number of variables,
$m$ is the number of variables per monomial, and
$N$ is the number of monomials (not counting the constant term).
The last two columns list the number of convolution and addition jobs.}
\label{tabtestpols}
\end{center}
\end{table}

To examine the scalability of our problem,
experiments are run for increasing degrees of truncation and 
for increasing levels of precision.
Does the shape of the test polynomials influence the execution times?

\subsection{Performance}

For each run, four times are reported.
The elapsed times of the kernel launches are measured by
{\tt cudaEventElapsedTime} and expressed in milliseconds.
The first two times are the sums of all elapsed times spent by all kernels,
respectively for all convolutions and all additions.
The third time is the sum of the first two times.
Each kernel launches involves the memory transfer of the index vectors
that define the coordinates of the jobs in the data arrays.
The fourth reported is the wall clock time which includes also
this memory transfer.  What is not included in the times is the
transfer of the input data arrays from the host to the device
and of the output data arrays from the device to the host.

Table~\ref{tabcards} summarizes execution times to evaluate
the first test polynomial~$p_1$ at a power series truncated
to degree~152 in deca double precision.
This degree is the largest one block of threads can manage
because of the limitation of the size of shared memory,
which is the same all five devices.

\begin{table}[hbt]
\begin{center}
\begin{tabular}{r||r|r|r|r|r}
     & \multicolumn{1}{c|}{C2050}
     & \multicolumn{1}{c|}{K20C}
     & \multicolumn{1}{c|}{P100}
     & \multicolumn{1}{c|}{V100}
     & \multicolumn{1}{c}{$\!\!$RTX 2080} \\ \hline \hline
$\!\!$convolution & 12947.26 & 11290.22 & 1060.03 & 634.29 & 10002.32 \\
 addition &    10.72 &    11.13 &    1.37 &   0.77 &     5.01 \\ \hline
 sum & 12957.98 & 11301.35 & 1061.40 & 635.05 & 10007.34 \\ \hline \hline
wall clock & 12964.00 & 11309.00 & 1066.00 & 640.00 & 10024.00 
\end{tabular}
\caption{Evaluating $p_1$ for degree $d = 152$ in deca double precision.
The last line is the wall clock time for all convolution and
addition kernels.  All units are milliseconds.}
\label{tabcards}
\end{center}
\end{table}

The ratio $12964/640 \approx 20.26$ is the speedup of the most
recent V100 over the oldest C2050.
Compare the ratio of the wall clock times for P100 over V100 
in Table~\ref{tabcards}: $1066/640 \approx 1.67$
with the ratios of theoretical double peak performance of the
V100 of the P100: $7.9/4.7 \approx 1.68$.

In about one second, the P100 performed 16,380 convolutions
and 9,084 additions (see Table~\ref{tabtestpols})
to evaluate and differentiate~$p_1$, at series truncated
at degree~$d = 152$ in deca double precision.
Following the counts in~\cite{Ver20}, one addition in deca double
precision requires 139 additions and 258 subtractions of doubles,
while one deca double multiplication requires 952 additions,
1743 subtractions, and 394 multiplications of doubles.
One convolution with zero insertion on series truncated at degree $d$
requires $(d+1)^2$ multiplications and $d(d+1)$ additions.
One addition of two series truncated at degree $d$ requires
$d+1$ additions.  So we have $16,380 (d+1)^2$ multiplications
and $16,380 d(d+1) + 9,084(d+1)$ additions in deca double precision.
One multiplication and one addition in deca double precision 
require respectively 3089 and 397 double operations.
Then the $16,380 (d+1)^2$ evaluates to 1,184,444,368,380
and $16,380 d(d+1) + 9,084(d+1)$ to 151,782,283,404
double float operations.  
In total, in 1.066 seconds the P100 performed 1,336,226,651,784
double float operations,
reaching a performance of about 1.25 TFLOPS.

Another observation from Table~\ref{tabcards} is the tiny
amount of time spent by the addition kernels, when compared
to the convolution kernels, for V100: 0.77 versus 634.29.
The convolution is quadratic in the degree~$d$, whereas
the addition is linear in~$d$.  As every block has $d+1$ threads, 
the addition finishes in one single step, whereas there are
still $d$ steps in the convolutions.

Table~\ref{tabp100} lists the execution times for~$p_2$
and~$p_3$ on P100 and V100, to verify if the shape of
the test polynomial would influence the conclusions on~$p_1$.

\begin{table}[hbt]
\begin{center}
\begin{tabular}{r||r|r||r|r}
     & \multicolumn{2}{c||}{$p_2$}
     & \multicolumn{2}{c}{$p_3$} \\ \cline{2-5}
     & \multicolumn{1}{c|}{P100}
     & \multicolumn{1}{c||}{V100}
     & \multicolumn{1}{c|}{P100}
     & \multicolumn{1}{c}{V100} \\ \hline \hline
 convolution & 1700.49 & 1115.03 & 1566.58 & 926.53 \\
 addition &       1.24 &    0.67 &    3.43 &   1.92 \\ \hline
 sum &         1701.72 & 1115.71 & 1570.01 & 928.45 \\ \hline \hline
wall clock &   1729.00 & 1142.00 & 1583.00 & 941.00
\end{tabular}
\caption{Evaluating $p_2$ and~$p_3$
for degree $d = 152$ in deca double precision.
The last line is the wall clock time for all convolution and
addition kernels.  All units are milliseconds.}
\label{tabp100}
\end{center}
\end{table}
The ratios of the wall clock times on P100 over the V100
for~$p_2$ and~$p_3$ are respectively
$1729/1142 \approx 1.51$ and $1583/941 \approx 1.68$.

For~$p_2$, the factor 1.51 is not as high as expected.
One probable cause is that the number of convolutions jobs 
in the first 31 layers equals 256.
This number equals the number of blocks in one kernel launch.
The number of streaming multiprocessors of the P100 and V100
respectively equal 56 and~80.
The number of 256 blocks in one launch does not occupy the V100
as much as the P100.

\subsection{Scalability}

The plots in this section visualize data in
Tables~\ref{tabV100runs1} and~\ref{tabV100runs2},
which contain times on the three test polynomials, on the V100.
These raw data sets are in the appendix.

As the times spent by all addition kernels is less than one
millisecond for~$p_1$, Figure~\ref{figaddp1} shows the relative cost
of the multiple doubles versus doubles.  
The cost starts to increase once the degrees become larger than the warp size.
For all precisions, the cost at degree 127 is less than twice
the cost at degree 63.

\begin{figure}[hbt]
\begin{center}
{\includegraphics[width=9cm]{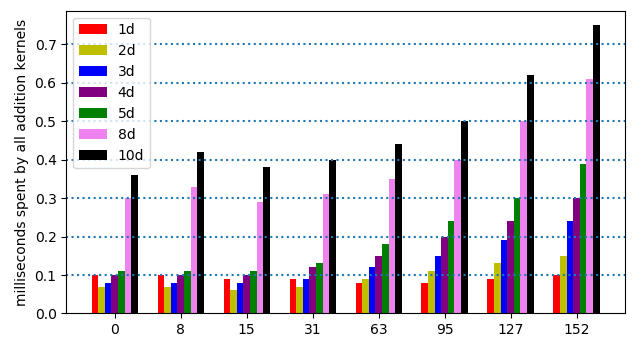}}
\caption{Times spent by all addition kernels when
evaluating~$p_1$ and its gradient at power series 
truncated at increasing degrees 0, 8, 15, 31, 63, 95, 127, and 152, 
for seven precisions: double (1d), 2d, 3d, 4d, 5d, 8d, and 10d,
at power series truncated to degree~191.}
\label{figaddp1}
\end{center}
\end{figure}

Figure~\ref{figaddp1p2p3} shows the sum of the times spent by all
addition kernels, for the three test polynomials,
for power series truncated at degree 152, for all seven precisions.
Although $p_3$ has 8,128 monomials and $p_2$ has only 128,
the increase in addition times for $p_3$ is at most three times as 
much as for~$p_2$.
The addition for $p_3$ happens with 12 kernel launches,
while the addition for $p_2$ has 7 kernel launches.

\begin{figure}[hbt]
\begin{center}
{\includegraphics[width=9cm]{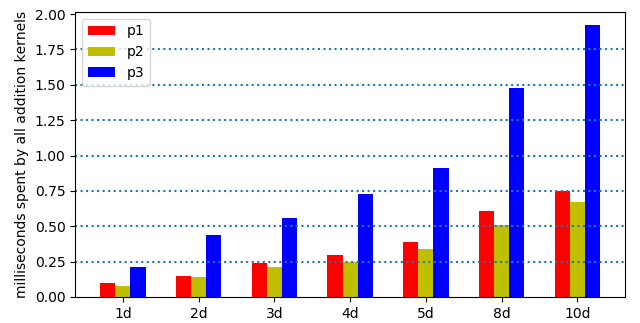}}
\caption{Times spent by all addition kernels when 
evaluating~$p_1$, $p_2$, $p_3$ and their gradients at
power series truncated at degrees~152,
for seven precisions: double (1d), 2d, 3d, 4d, 5d, 8d, and 10d.}
\label{figaddp1p2p3}
\end{center}
\end{figure}

The percentage of time spent by all kernels over the wall clock time
is visualized in Figure~\ref{figsumwall}.
For double precision, the wall clock time dominates 
(the percentage of the time spent by all kernels is less than 10\%),
although the time is also less than a tenth of a millisecond.
This percentage climbs for higher precisions. 
In triple precision, the time spent on all kernels dominates 
the wall clock time.  For octo and deca double precision, 
this percentage is more than~95\%.

\begin{figure}[hbt]
\begin{center}
{\includegraphics[width=9cm]{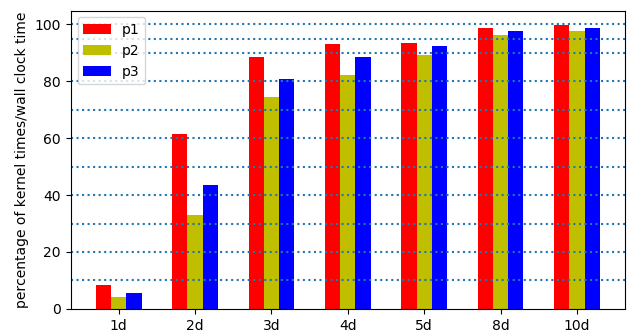}}
\caption{Percentage of the time spent by all kernels over
the wall clock time when evaluating~$p_1$, $p_2$, $p_3$ 
and their gradients at power series truncated at degrees~152,
for seven precisions: double (1d), 2d, 3d, 4d, 5d, 8d, and 10d.}
\label{figsumwall}
\end{center}
\end{figure}

As the precision increases, the problem becomes more and more compute bound.
For degree~191, in Table~\ref{tabV100runs1}, for $p_1$, the wall clock times
in double, double double, quad double, and octo double are respectively
6, 14, 95, and 449 seconds.  The cost overhead factor of double double
over double is typically a factor of about five, whereas here we observe
$14/6 \approx 2.33$.  The other observed cost overhead factors
are $95/12 \approx 6.79$ and $449/95 \approx 4.72$.
In Figure~\ref{figdoubles}, the evolution of the logarithmic
wall clock time is plotted.
 
\begin{figure}[hbt]
\begin{center}
{\includegraphics[width=9cm]{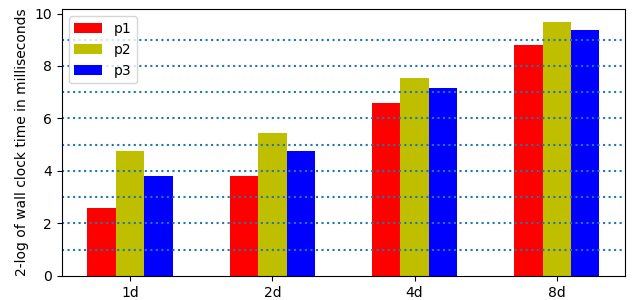}}
\caption{The 2-logarithm of the wall clock times
to evaluate and differentiate $p_1$, $p_2$, $p_3$ 
in double (1d), double double (2d), quad double (4d),
and octo double (8d) precision,
for power series truncated at degree 191.}
\label{figdoubles}
\end{center}
\end{figure}

If the number of coefficients in a truncated series doubles
from 32 to 64, and from 64 to 128, then one would expect
the observed wall clock times to quadruple, as the cost
of the convolutions is $O(d^2)$ for the truncation degree~$d$.
As shown in Figure~\ref{figscale}, the wall clock times doubles,
as the difference between the bars in the 2-log times is about one.
 
\begin{figure}[hbt]
\begin{center}
{\includegraphics[width=9cm]{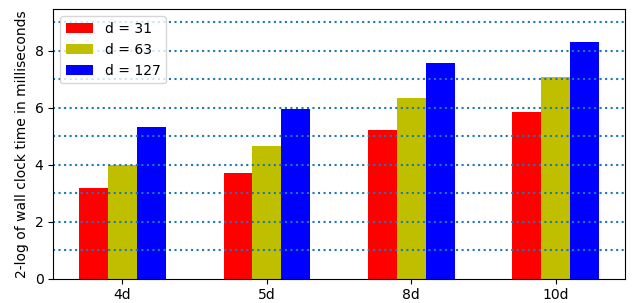}}
\caption{The 2-logarithm of the wall clock times
to evaluate and differentiate $p_1$ 
in quad double (4d), penta double (5d), octo double (8d),
and deca double (10d) precision,
for power series truncated at degrees 31, 63, and 127.}
\label{figscale}
\end{center}
\end{figure}

\section{Conclusions}

The evaluation and differentiation of a polynomial in $n$ variables
at power series truncated at some finite degree~$d$ 
requires a number of convolution jobs proportional
to the number of variables per monomial and the number $N$ of monomials.
The convolution jobs are arranged in layers of 
jobs that can be executed simultaneously.
A scan performs the $N$ addition jobs 
in $\lceil ~\! \log_2(N) ~\! \rceil$ steps.
For polynomials where $N$ dominates the number of variables per monomial,
The theoretical speedup is bounded by~$d N/\log_2(N)$.

Data staging algorithms define the coordinates for the convolution
and the addition jobs.
Speedup factors comparing the V100 and P100 are close to the
ratio of their theoretical peak performance.
Experimental results show that teraflop performance is obtained.
The accelerated algorithms scale well for increasing degrees and precisions.
GPUs are well suited to compensate for the overhead 
of power series arithmetic and multiple double precision.

\section*{Appendix}

The source code is available on github,
under the GNU GPL license,
as part of the code of PHCpack.
Tables~\ref{tabV100runs1}, \ref{tabV100runs2} and~\ref{tabV100runs3}
contain times on the three test polynomials, on the V100.
Table~\ref{tabV100wall} illustrates the fluctuation of the wall clock times.

\begin{table}[hbt]
\begin{center}
{\small
\begin{tabular}{r|r||rrrrrrrrrr}
\multicolumn{1}{c|}{~} 
    & \multicolumn{1}{c||}{$d$} 
    & \multicolumn{1}{c}{0}
    & \multicolumn{1}{c}{8} 
    & \multicolumn{1}{c}{15} 
    & \multicolumn{1}{c}{31} 
    & \multicolumn{1}{c}{63}
    & \multicolumn{1}{c}{95}
    & \multicolumn{1}{c}{127} 
    & \multicolumn{1}{c}{152}
    & \multicolumn{1}{c}{159}
    & \multicolumn{1}{c}{191} \\ \hline \hline
 1d & cnv & 0.08 & 0.07 & 0.07 & 0.07 & 0.11 & 0.17 & 0.28 & 0.39 & 0.40 & 0.56 \\
     & add & 0.10 & 0.10 & 0.09 & 0.09 & 0.08 & 0.08 & 0.09 & 0.10 & 0.10 & 0.11 \\
\hline
     & sum & 0.18 & 0.17 & 0.16 & 0.16 & 0.19 & 0.26 & 0.37 & 0.50 & 0.50 & 0.67 \\
\hline \hline
     & wall & 9.00 & 9.00 & 8.00 & 9.00 & 7.00 & 6.00 & 6.00 & 6.00 & 0.67 & 6.00 \\
\hline \hline
 2d & cnv & 0.06 & 0.11 & 0.17 & 0.31 & 0.98 &  2.39 & 3.58 &  7.20 &  7.48 &  9.23 \\
     & add & 0.07 & 0.07 & 0.06 & 0.07 & 0.09 &  0.11 & 0.13 &  0.15 &  0.16 &  0.18 \\
\hline
     & sum & 0.13 & 0.18 & 0.23 & 0.38 & 1.06 &  2.50 & 3.71 &  7.36 &  7.63 &  9.41 \\
\hline \hline
     & wall & 5.00 & 5.00 & 5.00 & 5.00 & 6.00 &  7.00 & 9.00 & 12.00 & 12.00 & 14.00 \\
\hline \hline
 3d & cnv & 0.10 & 0.57 & 1.00 & 2.00 &  5.80 & 13.82 & 19.88 & 38.70 & 40.53 & 52.03 \\
     & add & 0.08 & 0.08 & 0.08 & 0.09 &  0.12 &  0.15 &  0.19 &  0.24 &  0.22 &  0.26 \\
\hline
     & sum & 0.18 & 0.65 & 1.08 & 2.09 &  5.92 & 13.97 & 20.07 & 38.94 & 40.76 & 52.29 \\ 
\hline \hline
     & wall & 5.00 & 5.00 & 6.00 & 7.00 & 11.00 & 19.00 & 25.00 & 44.00 & 46.00 & 57.00 \\
\hline \hline
 4d & cnv & 0.15 & 1.24 & 2.19 & 4.39 & 11.01 & 23.99 & 35.40 & 65.76 & 68.51 & 90.40 \\
     & add & 0.10 & 0.10 & 0.10 & 0.12 &  0.15 &  0.20 &  0.24 &  0.30 &  0.29 &  0.33 \\
\hline
     & sum & 0.25 & 1.34 & 2.29 & 4.51 & 11.16 & 24.19 & 35.64 & 66.06 & 68.80 & 90.73 \\
\hline \hline
     & wall & 5.00 & 6.00 & 7.00 & 9.00 & 16.00 & 29.00 & 40.00 & 71.00 & 73.00 & 95.00 \\
\hline \hline
 5d & cnv & 0.25 & 2.23 & 3.98 &  7.94 & 20.59 & 42.87 & 57.19 & 114.57 & 111.68 & 143.70 \\
     & add & 0.11 & 0.11 & 0.11 &  0.13 &  0.18 &  0.24 &  0.30 &   0.39 &   0.36 &   0.42 \\
\hline
     & sum & 0.37 & 2.34 & 4.09 &  8.07 & 20.77 & 43.11 & 57.49 & 114.96 & 112.04 & 144.12 \\
\hline \hline
     & wall & 5.00 & 7.00 & 8.00 & 13.00 & 25.00 & 48.00 & 62.00 & 123.00 & 117.00 & 150.00 \\
\hline \hline
 8d & cnv & 0.82 &  8.92 & 15.97 & 32.26 & 77.24 & 150.64 & 182.09 & 359.68 & 377.88 & 442.90 \\
     & add & 0.30 &  0.33 &  0.29 &  0.31 &  0.35 &   0.40 &   0.50 &   0.61 &   0.59 &   0.67 \\
\hline
     & sum & 1.12 &  9.25 & 16.27 & 32.57 & 77.59 & 151.04 & 182.58 & 360.29 & 378.48 & 443.57 \\
\hline \hline
     & wall & 8.00 & 17.00 & 21.00 & 37.00 & 82.00 & 156.00 & 188.00 & 365.00 & 384.00 & 449.00 \\
\hline \hline
10d & cnv & 1.30 & 15.74 & 26.57 & 52.31 & 130.04 & 257.59 & 312.16 & 635.42 \\
     & add & 0.36 &  0.42 &  0.38 &  0.40 &   0.44 &   0.50 &   0.62 &   0.75 \\
\hline
     & sum & 1.66 & 16.16 & 26.95 & 52.71 & 130.48 & 258.09 & 312.78 & 636.17 \\
\hline \hline
     & wall & 7.00 & 30.00 & 35.00 & 58.00 & 135.00 & 263.00 & 317.00 & 641.00 \\
\end{tabular}
}
\caption{Times in milliseconds to evaluate and differentiate~$p_1$,
for increasing degree~$d$, and for increasing precision.}
\label{tabV100runs1}
\end{center}
\end{table}

\begin{table}[hbt]
\begin{center}
{\small
\begin{tabular}{c|r||rrrrrrrrrr}
\multicolumn{1}{c|}{~}
    & \multicolumn{1}{c||}{$d$}
    & \multicolumn{1}{c}{0}
    & \multicolumn{1}{c}{8} 
    & \multicolumn{1}{c}{15} 
    & \multicolumn{1}{c}{31} 
    & \multicolumn{1}{c}{63}
    & \multicolumn{1}{c}{95}
    & \multicolumn{1}{c}{127} 
    & \multicolumn{1}{c}{152}
    & \multicolumn{1}{c}{159}
    & \multicolumn{1}{c}{191} \\ \hline \hline
 1d & cnv &  0.41 &  0.41 &  0.42 &  0.43 &  0.50 &  0.63 &  0.80 &  1.01 &  1.04 &  1.32 \\
    & add &  0.05 &  0.05 &  0.05 &  0.05 &  0.05 &  0.05 &  0.06 &  0.08 &  0.08 &  0.08 \\
\hline
    & sum &  0.45 &  0.45 &  0.48 &  0.48 &  0.55 &  0.69 &  0.87 &  1.09 &  1.12 &  1.41 \\
\hline \hline
    & wall & 26.00 & 26.00 & 25.00 & 27.00 & 25.00 & 26.00 & 26.00 & 27.00 & 27.00 & 27.00 \\
\hline \hline
 2d & cnv &  0.42 &  0.55 &  0.69 &  1.01 &  2.42 &  4.87 &  6.84 & 12.35 & 12.89 & 16.19 \\
     & add &  0.05 &  0.05 &  0.05 &  0.05 &  0.07 &  0.09 &  0.11 &  0.14 &  0.13 &  0.15 \\
\hline
     & sum &  0.47 &  0.60 &  0.74 &  1.07 &  2.49 &  4.96 &  6.95 & 12.48 & 13.02 & 16.35 \\
\hline \hline
     & wall & 25.00 & 25.00 & 26.00 & 27.00 & 29.00 & 31.00 & 33.00 & 38.00 & 39.00 & 43.00 \\
\hline \hline
 3d & cnv &  0.53 &  1.53 &  2.44 &  4.50 & 11.71 & 24.59 & 34.53 &  75.74 &  78.59 & 94.57 \\
     & add &  0.06 &  0.06 &  0.06 &  0.07 &  0.09 &  0.13 &  0.16 &   0.21 &   0.20 &  0.22 \\
\hline 
     & sum &  0.58 &  1.59 &  2.51 &  4.58 & 11.80 & 24.72 & 34.69 &  75.95 &  78.78 & 94.79 \\
\hline \hline
     & wall & 27.00 & 28.00 & 29.00 & 31.00 & 37.00 & 50.00 & 61.00 & 102.00 & 105.00 & 120.00 \\
\hline \hline
 4d & cnv &  0.57 &  2.61 &  4.37 &  8.57 & 21.29 & 44.17 & 61.66 & 118.98 & 125.11 & 157.94 \\
     & add &  0.07 &  0.08 &  0.08 &  0.09 &  0.12 &  0.17 &  0.20 &   0.25 &   0.25 &   0.29 \\
\hline
     & sum &  0.65 &  2.68 &  4.45 &  8.66 & 21.41 & 44.34 & 61.87 & 119.23 & 125.37 & 158.23 \\ 
\hline \hline
     & wall & 26.00 & 29.00 & 31.00 & 35.00 & 48.00 & 70.00 & 87.00 & 145.00 & 151.00 & 184.00 \\ \hline \hline
 5d & cnv &  0.84 &  5.30 &  9.22 & 18.31 & 39.36 &  80.19 & 112.57 & 205.65 & 214.06 & 273.53 \\  
     & add &  0.09 &  0.09 &  0.10 &  0.11 &  0.15 &   0.20 &   0.25 &   0.34 &   0.31 &   0.36 \\
\hline
     & sum &  0.93 &  5.40 &  9.32 & 18.42 & 39.51 &  80.40 & 112.83 & 205.99 & 214.36 & 273.89 \\
\hline \hline
     & wall & 26.00 & 31.00 & 34.00 & 44.00 & 65.00 & 105.00 & 138.00 & 231.00 & 239.00 & 299.00 \\
\hline \hline
 8d & cnv &  1.76 & 16.56 & 29.58 & 59.66 & 139.71 & 253.36 & 328.69 & 639.72 & 672.51 & 789.62 \\
     & add &  0.23 &  0.24 &  0.25 &  0.26 &   0.30 &   0.35 &   0.42 &   0.51 &   0.51 &   0.58 \\
\hline
     & sum &  1.99 & 16.80 & 29.82 & 59.92 & 140.01 & 253.71 & 329.11 & 640.23 & 673.02 & 790.20 \\
\hline \hline
     & wall & 27.00 & 42.00 & 55.00 & 85.00 & 165.00 & 279.00 & 355.00 & 666.00 & 699.00 & 817.00 \\
\hline \hline
10d & cnv &  2.64 & 28.79 & 48.58 &  94.48 & 238.82 & 442.12 & 559.61 & 1115.03 \\
     & add &  0.29 &  0.31 &  0.32 &   0.34 &   0.38 &   0.45 &   0.54 &    0.67 \\
\hline
     & sum &  2.93 & 29.09 & 48.89 &  94.82 & 239.20 & 442.57 & 560.15 & 1115.71 \\
\hline \hline
     & wall & 29.00 & 55.00 & 75.00 & 120.00 & 265.00 & 468.00 & 586.00 & 1142.00 \\
\end{tabular}
}
\caption{Times in milliseconds to evaluate and differentiate $p_2$,
for increasing degree~$d$, and for increasing precision.}
\label{tabV100runs2}
\end{center}
\end{table}

\begin{table}[hbt]
\begin{center}
{\small
\begin{tabular}{c|r||rrrrrrrrrr}
    \multicolumn{1}{c|}{~}
    & \multicolumn{1}{c||}{$d$}
    & \multicolumn{1}{c}{0}
    & \multicolumn{1}{c}{8} 
    & \multicolumn{1}{c}{15} 
    & \multicolumn{1}{c}{31} 
    & \multicolumn{1}{c}{63}
    & \multicolumn{1}{c}{95}
    & \multicolumn{1}{c}{127} 
    & \multicolumn{1}{c}{152}
    & \multicolumn{1}{c}{159}
    & \multicolumn{1}{c}{191} \\ \hline \hline
 1d & cnv &  0.05 &  0.05 &  0.05 &  0.06 &  0.12 &  0.22 &  0.37 &  0.53 &  0.55 &  0.78 \\ 
     & add &  0.11 &  0.11 &  0.11 &  0.11 &  0.12 &  0.16 &  0.19 &  0.21 &  0.21 &  0.25 \\
\hline
     & sum &  0.16 &  0.15 &  0.15 &  0.17 &  0.24 &  0.37 &  0.55 &  0.74 &  0.77 &  1.03 \\
\hline \hline
     & wall & 12.00 & 13.00 & 12.00 & 12.00 & 13.00 & 13.00 & 13.00 & 13.00 & 14.00 & 14.00 \\
\hline \hline
 2d & cnv &  0.05 &  0.13 &  0.22 &  0.42 &  1.36 &  3.43 &  5.20 & 10.47 & 10.93 & 13.52 \\
     & add &  0.12 &  0.11 &  0.11 &  0.13 &  0.18 &  0.25 &  0.33 &  0.44 &  0.37 &  0.44 \\
\hline
     & sum &  0.17 &  0.24 &  0.34 &  0.54 &  1.54 &  3.69 &  5.52 & 10.91 & 11.30 & 13.96 \\
\hline \hline
     & wall & 13.00 & 13.00 & 13.00 & 13.00 & 14.00 & 17.00 & 18.00 & 25.00 & 24.00 & 27.00 \\
\hline \hline
 3d & cnv &  0.11 &  0.81 &  1.42 &  2.86 &  8.26 & 20.06 & 29.10 & 56.76 & 59.25 & 76.49 \\
     & add &  0.14 &  0.14 &  0.15 &  0.18 &  0.25 &  0.37 &  0.46 &  0.56 &  0.54 &  0.64 \\
\hline
     & sum &  0.25 &  0.95 &  1.57 &  3.04 &  8.52 & 20.43 & 29.56 & 57.32 & 59.79 & 77.13 \\
\hline \hline
     & wall & 13.00 & 14.00 & 14.00 & 16.00 & 21.00 & 33.00 & 43.00 & 71.00 & 73.00 & 90.00 \\
\hline \hline
 4d & cnv &  0.19 &  1.75 &  3.11 &  6.22 & 15.92 & 34.81 & 51.57 &  95.91 & 100.03 & 129.76 \\
     & add &  0.17 &  0.19 &  0.19 &  0.24 &  0.33 &  0.46 &  0.61 &   0.73 &   0.71 &   0.84 \\
\hline
     & sum &  0.36 &  1.94 &  3.30 &  6.45 & 16.25 & 35.27 & 52.18 &  96.64 & 100.75 & 130.61 \\
\hline \hline
     & wall & 13.00 & 14.00 & 16.00 & 19.00 & 29.00 & 49.00 & 65.00 & 109.00 & 114.00 & 144.00 \\ \hline \hline
 5d & cnv &  0.35 &  3.24 &  5.76 & 11.56 & 29.23 & 62.60 & 83.30 & 157.02 & 163.71 & 210.28 \\
     & add &  0.24 &  0.26 &  0.29 &  0.41 &  0.57 &  0.57 &  0.74 &   0.91 &   0.88 &   1.04 \\
\hline
     & sum &  0.59 &  3.50 &  6.02 & 11.84 & 29.63 & 84.04 & 84.04 & 157.93 & 164.59 & 211.31 \\
\hline \hline
     & wall & 15.00 & 17.00 & 18.00 & 24.00 & 43.00 & 76.00 & 97.00 & 171.00 & 178.00 & 224.00 \\
\hline \hline
 8d & cnv &  1.19 & 13.11 & 23.49 & 47.32 & 107.64 & 221.87 & 265.69 & 528.19 & 553.59 & 647.95 \\
     & add &  0.62 &  0.70 &  0.70 &  0.75 &   0.84 &   0.98 &   1.22 &   1.48 &   1.42 &   1.69 \\
\hline
     & sum &  1.80 & 13.80 & 24.18 & 48.07 & 108.48 & 222.84 & 266.31 & 529.67 & 555.01 & 649.64 \\
\hline \hline
     & wall & 14.00 & 27.00 & 37.00 & 61.00 & 121.00 & 236.00 & 280.00 & 542.00 & 573.00 & 663.00 \\
\hline \hline
10d & cnv &  1.90 & 23.12 & 39.12 & 75.81 & 181.99 & 380.19 & 455.78 & 926.53 \\
     & add &  0.80 &  0.88 &  0.89 &  0.94 &   1.04 &   1.19 &   1.47 &   1.92 \\
\hline
     & sum &  2.70 & 24.00 & 40.01 & 76.76 & 183.04 & 381.38 & 457.25 & 928.45 \\
\hline \hline
     & wall & 16.00 & 37.00 & 52.00 & 90.00 & 197.00 & 394.00 & 470.00 & 941.00
\end{tabular}
}
\caption{Times in milliseconds to evaluate and differentiate $p_3$,
for increasing degree~$d$, and for increasing precision.}
\label{tabV100runs3}
\end{center}
\end{table}

\begin{table}[hbt]
\begin{center}
\begin{tabular}{r|cccccc}
wall clock times & 941 & 942 & 943 & 944 & 945 & 946 \\ \hline \hline
fixed seed one   &  0  &  0  &  3  &  5  &  2  &  0  \\ \hline
different seeds  &  4  &  1  &  3  &  1  &  0  &  1
\end{tabular}
\caption{Wall clock times in milliseconds to evaluate and differentiate $p_3$
in deca double precision, for degree~152, with frequencies for ten runs,
once with the fixed seed one, and once with different seeds for the
random numbers.}
\label{tabV100wall}
\end{center}
\end{table}

\end{document}